\documentclass{article}
\usepackage[latin1]{inputenc}
 \usepackage{amsmath,amssymb,amsfonts,amsthm,fullpage}

\newtheorem{theorem}{Theorem} 
\newtheorem{lemma}{Lemma} 

\newcommand{\suppress}[1]{}
\newcommand{\comment}[1]{}

\newcommand{\eps}{\varepsilon}
\renewcommand{\epsilon}{\varepsilon}

\newcommand{\PLS}{\mathsf{PLS}}
\newcommand{\NP}{\mathsf{NP}}
\newcommand{\TFNP}{\mathsf{TFNP}}

\newcommand{\coNP}{\mathsf{coNP}}

\newcommand{\beq}{\begin{equation}}
\newcommand{\eeq}{\end{equation}}
\newcommand{\beqa}{\begin{eqnarray}}
\newcommand{\eeqa}{\end{eqnarray}}

\begin{document}

\title{The complexity of approximate Nash equilibrium \\
in congestion games with negative delays%
\thanks{Partially supported by the French ANR Defis program
under contract ANR-08-EMER-012 (QRAC project)
and
the European Commission IST STREP Project
Quantum Computer Science (QCS) 25596.}
}
\date{}

\author{
Fr\'ed\'eric Magniez\thanks{LIAFA, Univ. Paris 7, CNRS; F-75205 Paris, France.
\texttt{magniez@liafa.jussieu.fr}}
\and
Michel de Rougemont\thanks{Univ. Paris 2, and LIAFA, Univ. Paris 7, CNRS; F-75205 Paris, France.
\texttt{mdr@liafa.jussieu.fr}}
\and
Miklos Santha\thanks{LIAFA, Univ. Paris 7, CNRS;
F-75205 Paris, France;  and
Centre for Quantum Technologies,
National University of Singapore, Singapore 117543. \texttt{santha@lri.fr}}
\and
Xavier Zeitoun\thanks{Univ. Paris 11, and LIAFA, Univ. Paris 7, CNRS; F-75205 Paris, France.
\texttt{xavier.zeitoun@lri.fr}}
}

\maketitle

\begin{abstract}
We extend the study of the complexity of finding an $\eps$-approximate Nash equilibrium in  congestion games from the case of positive delay functions to delays of arbitrary sign. 
We first prove that  in symmetric games with $\alpha$-bounded jump the $\eps$-Nash dynamic  converges in polynomial time when all delay functions are negative, similarly to the case of positive delays.
We then establish a hardness result for symmetric games with $\alpha$-bounded jump and with arbitrary  delay functions: in that case finding an  $\eps$-Nash equilibrium becomes $\PLS$-complete.
\end{abstract}


\section{Introduction}
Congestion games were introduced by Rosenthal~\cite{R73} to model shared resources by selfish players.
In these games the strategies of each player correspond to some collection of subsets of a given set of common resources.
The cost of a strategy is the sum of the costs of the selected resources, where the cost of a particular resource depends on the 
number of players having chosen this resource. This dependence is described in the specification of
the game by a non-decreasing delay function for each resource.

Congestion games can describe several interesting routing and resource allocation scenarios in networks. 
More importantly from a game theoretic perspective, they have 
some particularly attractive properties. Rosenthal has proven that they belong to the class of
potential games where, for each player, an improvement (decrease) in his cost is reflected by an improvement
in a global
function, the potential function. 
This implies, in particular, that congestion games always have a pure Nash equilibrium.
More precisely, a Nash equilibrium can be reached by the so called Nash dynamics, in which an unsatisfied player switches his strategy
to a better one, which decreases his cost function. Since the same improvement is mirrored in the potential function, which can not be decreased
infinitely, this process indeed has to converge to an equilibrium in a finite number of steps. In an exact potential game the changes in the
individual cost functions and the potential function are not only identical in sign, but also in the exact value.
Monderer and Shapley~\cite{MS96} have proved
that congestion games and exact potential games are equivalent.

The existence of a potential function for congestion games allows us to cast searching
for a Nash equilibrium as a local search problem. The states, that is the strategy profiles of the players, 
are the feasible solutions,
and the neighborhood of a state consists of all authorized changes in the strategy of a single player. Then local
optima correspond to states where no player can improve individually his cost, that is exactly to Nash equilibria.
The potential of a state can be evaluated in polynomial time, and similarly
a neighboring state of lower potential can be exhibited, provided that there exists one. This means that the problem of finding a
Nash equilibrium in a congestion game belongs to the complexity class $\PLS$, Polynomial Local Search, 
defined in~\cite{JPY88, PY88}.
The class $\PLS$ is a subclass of $\TFNP$~\cite{MP91},  the family of $\NP$ search problems for which a solution is guaranteed to exist.
While $\PLS$ is not harder than $\NP \cap \coNP$, 
it is widely believed to be computationally intractable.
It contains several complete problems, such as the weighted 2SAT-FLIP, where one looks for
a truth assignment that maximizes the sum of the weights of satisfied clauses
among all assignments of Hamming distance one.
Fabrikant, Papadimitriou and Talwar~\cite{FPT04} have shown that finding a Nash equilibrium in congestion games is also $\PLS$-complete.
In addition, they have explicitly constructed games in which the Nash dynamics takes
exponential time to converge. It is worth to note that it is also highly unlikely
that computing a mixed Nash equilibrium in general games is feasible in polynomial time, even when the
number of players is restricted to two~\cite{DGP06,CD06}.

It is therefore natural to look for relaxed versions,
and in particular approximations,  of Nash equilibria which might be computed in polynomial time.
Approximate Nash equilibria of various games have been defined and studied both in the 
additive~\cite{KM02, LMM03, CD06b, DMP06, DMP07, FNS07,TS08} 
and in the multiplicative models of approximation~\cite{CS07, AL10}.  Here we consider
multiplicative $\eps$-approximate Nash equilibria, for $0 < \eps < 1$, that is 
states where no single player can improve his cost by more than a factor  
of $\eps$ by unilaterally changing his strategy. In this context,  the analogous concept
of the Nash dynamics is the $\eps$-Nash dynamics, where only $\eps$-moves are permitted, which 
improve the respective player's cost at least  by a factor of $\eps$. Rosenthal's potential function arguments
implies again that the $\eps$-Nash dynamics converges to an $\eps$-Nash equilibrium.

In a very interesting positive result, Chien and Sinclair~\cite{CS07} proved that 
in congestion games with three specific constraints
the $\eps$-Nash dynamics indeed does converge fast, in polynomial time.
The three constraints are positivity, symmetry, and $\alpha$-bounded jump.
The first two constraints are rather standard.
A congestion game is positive if all delay functions are non-negative, and it
is symmetric if
all players have the same strategies.
The third constraint puts a limit on the speed of growth of the delay functions. They define a positive congestion game 
to be with $\alpha$-bounded jump, for some $\alpha \geq 1$, if the delay functions can not grow  more than a factor $\alpha$
when their argument is increased by one. Their result states that in positive and symmetric congestion
games with $\alpha$-bounded jump, 
the $\eps$-Nash dynamics converges in polynomial time 
in the input length, $\alpha$ and $1/\eps$.

Could it be that the $\eps$-Nash dynamic converges fast in every congestion game? Skopalik
and V\"ocking have found a very strong evidence for the contrary. In a negative result~\cite{SV08}, they  
proved that for every polynomial time computable  $0 < \eps < 1$, finding an $\eps$-approximate Nash equilibrium
is $\PLS$-complete, that is just as hard as finding a Nash equilibrium.
In fact, they result is even stronger, it shows the $\PLS$-completeness
of the problem for positive games.



In this paper we extend these studies to the case when the delays can be also negative, that is some resources might
have the special status of improving the cost of the players when they are chosen. 
We first prove that in symmetric games with $\alpha$-bounded jump, when all delay functions are negative,  the $\eps$-Nash dynamic
converges in polynomial time, just as in the case of positive games. We then prove a hardness result: finding an  
$\eps$-Nash equilibrium in symmetric games with $\alpha$-bounded jump becomes  $\PLS$-complete
when  delay functions of arbitrary sign are allowed.
In fact, our result is somewhat stronger: the $\PLS$-completeness holds even
when all delay functions are of constant sign or when all the delays are of constant absolute value.


\section{Preliminaries and results}
\subsection{Context}
\label{prelim}
We recall the notions of congestion games, local search problems and
approximate Nash equilibrium.

\subsubsection{Congestion games}

For a natural number $n$, we denote by $[n]$ the set $\{1, \ldots,
n\}$. For an integer $n \geq 2,$ an {\em $n$-player game in normal
form}  is specified by a set of {\em (pure) strategies} $S_i$, and a
{\em cost} function $c_i : S \rightarrow \mathbb{Z},$ for each player $i \in [n]$, where $S = S_1 \times
\cdots \times S_n$ is the set of {\em states}. 
For $s \in S$, the value $c_i(s)$ is the cost of player $i$ for 
state $s$. 
A game is {\em symmetric} if $S_1 = \ldots = S_n$. 

For a state $s = (s_1, \ldots, , s_n) \in S$, and for a pure strategy $t_i \in S_i$, we let 
$(s_{-i}, t)$ to be the state 
\mbox{$(s_1, \ldots , s_{i-1}, t, s_{i+1}, \ldots, s_n) \in S$}.
A {\em pure Nash equilibrium} is a state $s$ such that for all $i$, and for all pure strategies 
$t \in S_i$, we have 
$$c_i(s) \leq c_i(s_{-i}, t).$$
In general  games do not necessarily have a pure 
Nash-equilibrium.

A specific class of games which always have a pure Nash equilibrium are {\em congestion games},
where the cost functions are determined by the shared use of resources. More precisely, an
$n$-player congestion game is a 4-tuple
$G = (n, E, (d_e)_{e \in E},(S_i)_{i \in [n]})$, where $E$ is a finite set  of
{\em edges} (the common resources), 
$d_e : [n] \rightarrow \mathbb{Z}$ is a non-decreasing 
{\em delay function}, for every $e \in E$, and
$S_i \subseteq 2^E$ is the set of pure strategies of player $i$, for $i \in [n]$.
Given a state
$s = (s_1, \ldots , s_{i},  \ldots, s_n)$, let the {\em congestion} of $e$ in $s$ be
$f_e(s) = |\{ i \in [n] : e \in s_i \} |$. The cost function of user $i$ is defined then as 
$c_i(s) = \sum_{e \in s_i} d_e(f_e(s))$. Intuitively, each player uses some set of
resources, and the cost of each resource $e$ depends on the number of players using it, as described by the
delay function. To simplify the notation, we will specify a symmetric congestion game by 
a 4-tuple $G = (n, E, (d_e)_{e \in E},Z)$, where by definition the set of pure strategies of every player is 
$Z\subseteq 2^E$. 
We will refer to $Z$ as the set of {\em available} strategies.

That congestion games have indeed a Nash equilibrium can be easily shown by a potential function argument,
due to Rosenthal~\cite{R73}, as follows.
Let us define the potential function $\phi$ on the set of states as
$$\phi(s) = \sum_{e \in E} \sum_{t=1}^{f_e(s)} d_e(t).$$
If $s = (s_1, \ldots , s_{i},  \ldots, s_n)$ and  $s'=(s_{-i},s_i')$
are two states 
differing only for player $i$ then
$\phi(s) - \phi(s') = c_i(s) - c_i(s')$ since both of these quantities are in fact equal to
$\sum_{e \in s_i \setminus s'_i} d_e (f_e(s)) - \sum_{e \in s'_i \setminus s_i} d_e (f_e(s')).$
Therefore, in any state which is not a pure Nash equilibrium, there is always a player that can change unilaterally his strategy so that the induced new state has smaller potential. In fact the decrease in the cost function and in the potential are identical. This means that a finite
sequence of such individual changes, the so-called {\em Nash dynamics}, necessarily results in a pure Nash equilibrium 
since the integer valued potential function can not decrease forever. Therefore  congestion games can be
casted as local search problems, and the finding of a Nash equilibrium can be interpreted as the search of a local optimum.

\subsubsection{Local search problems}

A local search problem  is defined by a 4-tuple $\Pi=(\mathcal{I} , F ,(v_I)_{I\in \mathcal{I}}, (N_I)_{I\in \mathcal{I}})$, 
where $\mathcal{I}$ the set of instances,
$F$ maps every instance $I\in\mathcal{I}$ to a finite set of feasible solutions $F(I)$,
the objective function $v_I: F(I) \rightarrow \mathbb{Z}$ gives the value $v_I(S)$ of a feasible solution,
and $N_I(S) \subseteq F(I)$  is the neighborhood of $S \in F(I)$. Given an instance $I$,
the goal is to find a feasible solution $S \in F(I)$ such that is also local minimum, that is
for all $S'\in N_I(S)$, it satisfies $v_I(S) \leq v_I(S')$. 
A local search problem is in the class $\PLS$~\cite{JPY88, PY88} if there exist polynomial algorithms 
in the instance length to compute: an initial solution $S_0$;
the membership in $F(I)$;
the objective value $v_I(S)$; 
and a feasible solution $S' \in N_I(S)$ such that $v_I(S') < v_I(S)$ whenever $S$ is not a local minimum.
Computing a Nash equilibrium of congestion games is then indeed in PLS: Given an instance $G$,
the feasible solutions $F(G)$ are the states $S$, the value $v_G(s)$ of a state $s$ is its potential $\phi(s)$, and the 
neighborhood $N_G(s)$ consists of those states which differ in one coordinate from $s$.

The notion of $\PLS$-reducibility was introduced in \cite{JPY88}. A problem
$\Pi=(\mathcal{I} , F ,(v_I)_{I\in \mathcal{I}}, (N_I)_{I\in \mathcal{I}})$ is $\PLS$-reducible to 
$\Pi'=(\mathcal{I}' , F' ,(v'_{I})_{I\in \mathcal{I'}}, (N_I')_{I\in \mathcal{I'}})$ if there exist polynomial time computable functions 
$f:  \mathcal{I} \rightarrow  \mathcal{I}'$ and $g_I : F(f( I)) \rightarrow F(I)$, for ${I\in \mathcal{I}}$,
such that  if $S'$ is a local optimum of $f(I)$ then $g_I(S')$ is local optimum of $I$.  
Complete problems in $\PLS$ are not believed to be solvable by efficient procedures.
%
%
Therefore, it is highly unlikely that there exists at all a polynomial time algorithm for 
finding a pure equilibrium in congestion games. Indeed,
Fabrikant, Papadimitriou and Talwar~\cite{FPT04} have shown that this problem is
$\PLS$-complete, even for symmetric games.

\subsubsection{Approximate Nash equilibirum}

Several relaxations of the notion of equilibrium have been
considered in the form of 
approximations. Let $0 < \eps < 1$.
In our context $\eps$ will be a constant or some polynomial time computable function in the input length.
A {\em multiplicative $\eps$-approximate Nash equilibrium} is a 
state $s$ such that for all $i \in [n]$, and for all
strategies $t \in S_i$, we have 
$$c_i(s) - c_i(s_{-i}, t) \leq \eps |c_i(s)|.$$
Given a state $s$ and a strategy $t \in S_i$, we say that $(s_{-i}, t)$ is an {\em $\eps$-move} for player $i$ if
$$c_i(s) - c_i(s_{-i}, t) > \eps |c_i(s)|.$$
Clearly $s$ is an $\eps$-approximate Nash equilibrium if no player has an $\eps$-move.

The {\em $\eps$-Nash dynamics} is defined as a sequence of $\eps$-moves, where a player with 
the 
\textbf{largest absolute gain}
 makes the change in his strategy, when several players with $\eps$-move are available. Analogously to the exact
case, the $\eps$-Nash dynamics converges to an $\eps$-approximate Nash equilibrium. Finding an $\eps$-approximate Nash equilibrium
is also a problem in $\PLS$. When casting this as a local search, the only difference with the exact equilibrium case
is that the neighborhoods are restricted to states which are reachable by an $\eps$-move.

\subsection{
Related results}
In ~\cite{CS07} Chien and Sinclair have considered the rate of convergence of the $\eps$-Nash dynamics in 
symmetric congestion games
with two additional restrictions on the delay functions. 
We say that a delay function $d_e$ is positive if 
the delays
$d_e(t)$ are non-negative integers for all $1 \leq t \leq n$. 
A congestion game is {\em positive} if all delay functions are positive.
Let $\alpha \geq 1$.
A positive delay function is 
with $\alpha$-bounded jump if
the delays satisfy $d_e(t+1) \leq \alpha d_e(t)$, for all
$t \geq 1$. 
 We can think of $\alpha$ as being a constant,
or a polynomial time computable function in the input length of the game.
Obviously, a positive delay function with $\alpha$-bounded jump 
can never take the value 0.
A positive game is {\em with $\alpha$-bounded jump} if
all delay functions are with $\alpha$-bounded jump.
Chien  and Sinclair have shown that in symmetric positive games with bounded jump the
$\eps$-Nash dynamics converges in polynomial time. 

\begin{theorem}[Chien and Sinclair~\cite{CS07}]
\label{theorem:sinclair}

For every $\alpha \geq 1$ and $0 < \eps < 1$, in  $n$-player symmetric, positive congestion games 
with $\alpha$-bounded jump the $\eps$-Nash
dynamics converges  from any initial state in $O(n\alpha {\epsilon^{-1}}\log(nmD))$ steps, 
where $m = |E|$, and $D = \max \{  d_e(n) : e \in E \}$ is an upper bound on the delay functions.
\end{theorem}

The hope that the $\eps$-Nash dynamics  
converges  fast in generic congestion games was crushed by  Skopalik and V\"ocking~\cite{SV08}. In a strongly negative result they  
proved that 
finding an $\eps$-approximate Nash equilibrium
is $\PLS$-complete, that 
is just as hard as finding a Nash equilibrium. In fact, they result is even stronger, it shows the $\PLS$-completeness
of the problem also in positive games. 

\begin{theorem}[Skopalik and V\"ocking~\cite{SV08}]
\label{theorem:skopalik}
For every polynomial time computable  $0 < \eps <1$, 
finding an $\eps$-approximate Nash equilibrium in a positive congestion game is $\PLS$-complete.
\end{theorem}

\subsection{Our contributions}

In this paper we study the complexity of finding an $\eps$-approximate Nash equilibrium in congestion games 
where the delay functions can also have negative values. Our reason to study negative delays is twofold.
Firstly, negative delays are motivated by real scenarios worth of investigations. They correspond to situations where
some edges are encouraged to be taken by an authority, for example for regulation purposes. Such negative edges become incentives, 
whereas the positive edges carry the traditional penalties. Secondly, games with delays of arbitrary sign
seem to have a richer mathematical structure than games with only positive delays. Indeed, we will show that 
in symmetric games with $\alpha$-bounded jump, while finding 
an $\eps$-approximate Nash equilibrium is easy when delays can only be negative, just as in positive games,
the problem becomes significantly harder when they can be of arbitrary sign.

We feel that the case of the negative delays is somewhat analogous to the case of edges with negative weights in shortest paths
problems. Negative weights behave in that context also  as incentives, and the structure of shortest path problems is more complex
when negative weights are also allowed. For example, Dijkstra's algorithm for single source shortest paths works correctly
only in the case of  edges with positive weight.

In Section~\ref{sec-walk-search} we deal with symmetric games with $\alpha$-bounded jump, where all delay functions are negative.
In {\bf Theorem~\ref{theorem:main}} we show that the $\eps$-Nash dynamics converges in polynomial time in the input length, $\alpha$ and $1/\eps$. This result is analogous to the result of Chien and Sinclair~\cite{CS07}, stated in Theorem~\ref{theorem:sinclair},
but its proof is significantly harder. In positive games the proof essentially shows that while the equilibrium state is not reached,
one can always find a player whose cost is polynomially related to the potential, and who either can make an $\eps$-move himself,
or whose cost can lower bound the gain of any player 
with an $\eps$-move. In negative games it is not always possible to find a player with such a high cost, 
because we can relate the cost functions  to the
appropriate potential only when it is  restricted to edges with non-trivial congestion. 
To deal with the unaccounted edges, we show that 
we can find an edge whose initial delay is polynomially related to the remaining part of the potential.
We then make a somewhat subtle case analysis which considers also strategies
involving these edges.

In Section~\ref{hard} we extend our investigation to symmetric games with $\alpha$-bounded jump, and
with arbitrary delay functions. We show that in that case the problem of finding an $\eps$-approximate Nash equilibrium becomes
$\PLS$-complete. In fact, we can show this even with some specific restrictions on the delays. In {\bf Theorem~\ref{theorem:non-alternating}} we prove
the hardness result when all delay functions are of constant sign, and in {\bf Theorem~\ref{theorem:flip}} when delays can
change their sign, but remain of constant absolute value. For the $\PLS$-reductions we use the $\PLS$-complete problem of
Skoplik and V\"ocking~\cite{SV08} stated in Theorem~\ref{theorem:skopalik}.

\section{Negative games}
\label{sec-walk-search}
We start now the study of finding $\eps$-approximate Nash equilibria in congestion
games where the delay functions can take negative values. In this section we impose the restriction that the delay
functions have only negative values. We further suppose that the games are symmetric and $\alpha$-bounded.
We show in a result analogous to Theorem~\ref{theorem:sinclair} that for any polynomial time computable $\alpha$ and $\eps$, 
the $\eps$-Nash dynamics converges in polynomial time.

We say that a delay function $d_e$ is negative if 
the delays
$d_e(t)$ are negative integers for all $1 \leq t \leq n$. 
A congestion game is {\em negative} if all delay functions are negative.
Let $\alpha \geq 1$. 
A negative delay function is 
with $\alpha$-bounded jump if
the delays satisfy $d_e(t+1) \leq  d_e(t)/\alpha$, for all
$t \geq 1$.
A negative game is {\em with $\alpha$-bounded jump} if
all delay functions are with $\alpha$-bounded jump.

\begin{theorem}
\label{theorem:main}
For every $\alpha \geq 1$,
in an $n$-player symmetric, negative congestion game with $\alpha$-bounded jump the Nash
dynamics converges  from any initial state in $O( (\alpha n^2 +nm) \epsilon^{-1} \log (nmD))$ steps
where $m = |E|$, and  $D = \max \{ - d_e(1) : e \in E \}$ is an upper bound on magnitude of the delay functions.
\end{theorem}

\begin{proof}
We will suppose without loss of generality that every edge appears in some strategy,
since otherwise the edge can be discarded from $E$. 
We first define a positive potential function which will be appropriate to measure the
progress of the $\eps$-Nash dynamics. Let $\psi$ be defined over the states as 
$\psi(s) = - \sum_{e \in E} \sum_{t= f_e(s) + 1}^{n} d_e(t)$. 
The function $\psi$ is clearly positive, and we claim that it is a potential function, 
that is $\psi(s) - \psi(s') = c_i(s) - c_i(s')$
if the states $s$ and $s'$ differ only in their $i$th coordinate. This follows immediately from the fact
that for every state $s$, we have 
$\psi(s) = \phi(s) -k$, where $\phi(s) = \sum_{e \in E} \sum_{t=1}^{f_e(s)} d_e(t)$
is the Rosenthal potential function, and $k$ is the constant $- \sum_{e \in E} \sum_{t=1}^{n} d_e(t).$
Observe that $\psi(s)$ is bounded from above by $nmD$, for every state $s$.

For an arbitrary initial state $s^{(0)}$, let $s^{(k)}$ be the $k$th state of the $\epsilon$-Nash dynamics process.
We claim that $\psi(s^{(k+1)})  \leq 
\psi(s^{(k)}) (1 - \epsilon / 4(\alpha n^2 +nm))$, for every $k$, which clearly implies the theorem.
Suppose that $s^{(k)} = s = (s_1, \ldots, s_n)$ is not an $\epsilon$-equilibrium, and 
let $i$ be the player which can make the largest gain $\epsilon$-move.
To prove our claim, we will show that there exists a strategy $s_i'$ for player $i$ such that
$c_i(s) - c_i(s_{-i}, s_i') \geq \psi(s)  \epsilon / 4(\alpha n^2 +nm)$.
Hence, an $\eps$-move of the $\eps$-Nash dynamics can only be better than
strategy $s_i'$ for player $i$.

The first idea is to try to prove, analogously to the case of positive games, that for some player $j$, the opposite of its cost $-c_j(s)$ is a polynomial
fraction of $\psi(s)$. Unfortunately this is not necessarily true. The sum $\sum_{j=1}^n c_j(s)$ is not necessarily a polynomial fraction of
$\psi(s)$ because edges whose congestion is 0 in $s$ do not contribute to the former, but do contribute the latter.
Therefore we introduce the function $\psi'$ as $\psi$
restricted to the edges with nontrivial congestion, that is by definition
$\psi'(s) = - \sum_{f_e(s) \neq 0} \sum_{t= f_e(s) + 1}^{n} d_e(t)$. 
The following Lemma shows that some of the $-c_j(s)$ is at least a polynomial fraction of $\psi'(s)$.

\begin{lemma}
\label{lem:fraction}
There exists a player $j$ such that
$$ -c_j(s) \geq  \psi'(s)/n^2.$$
\end{lemma}
\begin{proof}
We claim that 
\begin{equation*}
\label{eqn:big}
- n \sum_{j=1}^n c_j(s) \geq \psi'(s),
\end{equation*}
from which the statement  clearly follows.
To prove the claim 
we proceed by the following series of (in)equalities:
\begin{eqnarray*}
- n \sum_{j=1}^n c_j(s) & = & - n \sum_{f_e(s) \neq 0} f_e(s) ~d_e(f_e(s))\\
& \geq & - n \sum_{f_e(s) \neq 0} d_e(f_e(s))\\
& \geq & - \sum_{f_e(s) \neq 0} \sum_{t= f_e(s) + 1}^{n} d_e(t)\\
& =  &\psi'(s),
\end{eqnarray*}
where the second inequality holds because the delay functions are
non-decreasing. 
\end{proof}
We fix a value $j$ which satisfies Lemma~\ref{lem:fraction} for the rest of the proof.
To upper bound $\psi(s)$, we also have to consider the edges of congestion 0, besides the edges which are
accounted for in $\psi'(s)$. 
We have
$$\psi'(s) - n \sum_{f_e(s) = 0} d_e(1) \geq \psi(s),$$
again because the delays are non-decreasing.
This implies that either $\psi'(s) \geq \psi(s)/2$ or $- n \sum_{f_e(s) = 0} d_e(1) \geq \psi(s)/2$,
and the proof proceeds by distinguishing these two cases.

Case 1: $\psi'(s) \geq \psi(s)/2$. We then reason in two sub-cases by comparing
the value of $c_i(s)$ to  $\psi'(s) / 2 \alpha n^2$.
If $ - c_i(s) \geq \psi'(s) / 2 \alpha n^2$, then let $s_i'$ be the strategy
which makes the biggest gain for player $i$. 
Then we have 
\begin{eqnarray*}
c_i(s) - c_i(s_{-i}, s_i'))
    & \geq & - \epsilon c_i(s) \\
    & \geq &  \epsilon \psi(s) / 4 \alpha n^2,
\end{eqnarray*}
where first inequality holds since the move of player $i$ is an $\epsilon$-move,
and the second inequality is true because of the hypotheses.
If $ - c_i(s) < \psi'(s) / 2 \alpha n^2$, then let $s_i' = s_j,$ the strategy of player $j$ in state $s$.
Observe that $s_j$ is an available strategy for player $i$ since the game is symmetric. Then
\begin{eqnarray*}
c_i(s) - c_i(s_{-i}, s_i'))
    & \geq & c_i(s) - c_j(s)/\alpha \\
    & \geq &  \psi'(s) /  \alpha n^2  -  \psi'(s) / 2 \alpha n^2\\
    & \geq & \psi(s) / 4 \alpha n^2.
\end{eqnarray*}
Here the first inequality is true because the game is with $\alpha$-bounded jump. The second inequality 
follows from  the hypothesis and because $ c_j(s) \geq \psi'(s)/n^2.$ Finally, the third inequality holds
because $\psi'(s) \geq \psi(s)/2$.

Case 2: $- n \sum_{f_e(s) = 0} d_e(1) \geq \psi(s)/2$. Then for some edge with $f_e(s) = 0$, we
have $-d_e(1) \geq \psi(s)/2nm$. Let's fix such an edge $e$. We distinguish two sub-cases now 
by comparing
the value of $c_i(s)$ to  $d_e(1)/2$.  If $c_i(s) \leq d_e(1)/2$ then 
let $s_i'$ be the strategy
which makes the biggest gain for player $i$. 
Then, similarly to the first sub-case of Case 1, using the hypotheses and that 
player $i$'s move is an $\epsilon$-move,
we have 
\begin{eqnarray*}
c_i(s) - c_i(s_{-i}, s_i'))
    & \geq & - \epsilon c_i(s) \\
    & \geq &  \epsilon \psi(s) / 4nm.
\end{eqnarray*}
If $c_i(s) > d_e(1)/2$ then let $s'_i$ be some strategy that contains the edge $e$. There exists
such a strategy since useless edges were discarded from $E$. Then
$f_e(s_{-i}, s_i')) =1$ since $f_e(s) =0$ and $s$ and $(s_{-i}, s_i')$ differ only for the $i$th player.
This, in turn, implies that $c_i(s_{-i}, s_i')) \leq d_e(1)$, since the delays are negative. Therefore
\begin{eqnarray*}
c_i(s) - c_i(s_{-i}, s_i'))
    & \geq & c_i(s) - d_e(1) \\
    & \geq &  - d_e(1)/2\\
    & \geq & \psi(s) / 4 nm,
\end{eqnarray*}
where the last two inequalities follow from the hypotheses.

\end{proof}

\section{Games without sign restriction}
\label{hard}
In this section we deal with congestion games with no restriction on the sign of the delay functions.
Our overall result is that  in that case computing an $\eps$-approximate Nash equilibrium is $\PLS$-hard,
even when the two restrictions of Chien and Sinclair are kept, that is 
when the game is symmetric and with $\alpha$-bounded jump, for $\alpha \geq 1$.  Observe that the smaller
$\alpha$ the stronger is the hardness result, therefore we deal only with constant $\alpha$. Our first step is to observe that a 
simple consequence of Theorem~\ref{theorem:skopalik} is that finding an $\eps$-approximate Nash equilibrium in positive games
remains $\PLS$-complete even if we additionally suppose that the game is symmetric. Our reductions will use the
hardness of this latter problem. The proof of this statement is a $\PLS$-reduction of the search
of an $\eps$-approximate Nash equilibrium in positive games to the same problem in symmetric and positive games.
This reduction is basically identical
to the analogous reduction for pure Nash equilibria, due to Fabrikant, Papdimitriou and Talwar~\cite{FPT04}. We include here the
proof just for the sake of completeness.

\begin{theorem}
\label{theorem:fabrikant}
For every polynomial time computable  $0 < \eps <1$, 
finding an $\eps$-approximate Nash equilibrium in a symmetric and positive congestion game is $\PLS$-complete.
\end{theorem}

\begin{proof}
Given a congestion game $G$ with edge set $E$ and strategy sets $S_1, \ldots, S_n$, 
we map it to the symmetric game $G'$ defined as follows. 
The edge set of $G'$ is $E \cup \{e_1, \ldots , e_n\}$ where the $e_i$'s are new edges.
The set of available common strategies is $\bigcup_{i=1}^n S'_i$ where $S'_i= \{s \cup \{e_i\} :s \in S_i\}$. 
Set $D = \sum_{e\in E} d_e(n)$.
The delays of the edges in $E$ don't change, and for every $i \in [n]$, the delay   of $e_i$
is defined as
$$ 
d_{e_{i}}(t) = \left\{
    \begin{array}{ll}
       0 &\mbox{if } t=1, \\
       {D}/(1-\eps) & \mbox{if } t\geq 2.
    \end{array}
\right.	
$$

Observe that $D$ is an upper bound on the cost of the players in any state of $G$.  Let $s' = (s'_1,...,s'_n)$ be an 
$\eps$-approximate Nash equilibrium in $G'$. For every $i$, there is necessarily a unique $j_i$ such that
$s'_{j_i} \in S_i$. Indeed, if the strategies of several players would belong to the same $S_i$ then any of these players could pick
a strategy in which the congestion of the new edge would be 1, and therefore its cost would drop by at least a factor of $\eps$.
It is then  immediate that the state $s = (s_1, \ldots, s_n)$, where by definition $s_i = s'_{j_i} \setminus \{e_i\}$, is an
$\eps$-approximate Nash equilibrium in $G$.
\end{proof}

We need to discuss now the right notion of $\alpha$-bounded jump when the jump occurs
from a negative to a positive value in the delay function. One possibility could be to require
$d_e(t+1) \leq -\alpha d_e(t)$  when  $d_e(t)<0$  and  $d_e(t+1) )\geq 0$, but there are also
other plausible definitions. In fact, we will avoid to give a general definition because it turns out
that this is not necessary for our hardness results. Indeed, we will be able to establish a hardness result
for congestion games where there is no jump at all around 0, that is for 
delay functions of constant sign (still some of the delay functions can be negative while some others positive).
We say that a congestion game is  {\em non-alternating}, if 
every delay function is positive or negative. 
Let $\alpha > 1$ be a constant. A non-alternating congestion game is 
{\em with $\alpha$-bounded jump} if
all delay functions are with $\alpha$-bounded jump. 

What happens when $\alpha = 1$? 
We could also consider non-alternating games with 1-bounded jump, but they are not interesting.
These are games with constant delay functions for which a pure Nash equilibrium can be determined trivially.
Indeed, the cost functions of the individual players are independent from the strategies of the other players,
and therefore any choice of a least expensive strategy,  for each player, forms a Nash equilibrium.
Nonetheless, if we authorize a jump around 0, then even if the jump changes only the sign without changing the
absolute value (which corresponds intuitively to the case $\alpha = 1$ in that situation), the game becomes already hard.
We say that a delay function $d_e$ is a {\em flip function}, if there exists a positive integer $c $ such that
for some $1 \leq k \leq n$, the function satisfies

$$ d_{e}(t) = \left\{
    \begin{array}{ll}
        -c & \mbox{~if ~~} t<k, \\
       c & \mbox{~if ~~} t\geq k.
    \end{array}
\right.	$$
Flip functions are either constant positive functions, or they are 
simple step functions, which are
constant negative up to some point, where an alternation occurs which keeps the absolute value.
After the alternation the function remains constant positive. 
A congestion game is a {\em flip} game if all delay functions are flip functions.
The next two theorems state our hardness results respectively for non-alternating games with
$\alpha$-bounded jump and for flip games. 
\begin{theorem}
\label{theorem:non-alternating}
For every constant $\alpha >1$, and for every polynomial time computable $0 < \eps < 1$,
computing an $\epsilon$-approximate Nash equilibrium in $n$-player
symmetric, non-alternating congestion games 
with $\alpha$-bounded jump is $\PLS$-hard.
\end{theorem}

\begin{proof}
As stated in Theorem~\ref{theorem:fabrikant} finding an $\eps$-approximate Nash equilibrium in a symmetric and positive congestion game 
is $\PLS$-complete \cite{SV08}. 
We present a $\PLS$-reduction from this problem to the problem of finding 
an $\eps$-approximate Nash equilibrium in a symmetric and non-alternating game
with $\alpha$-bounded jump.

Let $G = (n, E, (d_e)_{e \in E}, Z) $ a symmetric and positive congestion game,
and let $\alpha >1$ be a constant. In our reduction we map $G$ to the symmetric game
$G' = (n, E', (d_{e'})_{e' \in E'}, Z') $ that we define now.
For each $e \in E$, we set 
$E_e=\{ e_1 ,e_{2}^+,e_{2}^-, \ldots , e_{n}^+, e_{n}^- \},$ 
and for every $z \subseteq E$, we define $z' = \bigcup_{e\in z} E_e$ (and therefore
$E' = \bigcup_{e\in E} E_e).$
The set of available strategies is defined as 
$Z' = \{z' : z \in Z\}.$
Finally the delay functions are defined as follows. The delay $d_{e_1}$ is 
simply the constant function  $d_{e}(1)$. For $k \geq 2$, we set

$$ d_{e_{k}^+}(t) = \left\{
    \begin{array}{ll}
       ( d_e(k) - d_e(k-1) ) \frac{\alpha}{\alpha^2-1}  & \mbox{if } t<k, \\
        (d_e(k)-d_e(k-1) ) \frac{\alpha^2}{\alpha^2-1} & \mbox{if } t\geq k,
    \end{array}
\right.	$$	
	and
	
$$ d_{e_{k}^-}(t) = \left\{
    \begin{array}{ll}
        -(d_e(k)-d_e(k-1)) \frac{\alpha}{\alpha^2-1}  & \mbox{if } t<k, \\
        -(d_e(k)-d_e(k-1))\frac{1}{\alpha^2-1} & \mbox{if } t\geq k.
    \end{array}
\right.	
$$
The game $G'$ is clearly non-alternating and with $\alpha$-bounded jump.

Observe that there is a bijection between the states of $G$ and $G'$. Indeed,
the states of $G'$ are of the form $s' = (s_1', \ldots , s_n')$, where 
$s=(s_1, \ldots , s_n) \in Z^n$ is a state of $G$. For the reduction we will simply show that if $s'$ is an
$\epsilon$-approximate Nash equilibrium in $G'$ then $s$ is an $\epsilon$-approximate 
Nash equilibrium in $G$ (our 
construction satisfies also the reverse implication). In fact, we show a stronger statement about cost functions:
for every state $s$, and for every player $i$, the cost of player $i$ for $s$ in $G$ is the same as the cost of
player $i$ for $s'$ in $G'$.

The edges $e_k^+$ and $e_k^-$ are such that the sum of their delay functions emulates the jump $d_e(k)-d_e(k-1)$
when $t \geq k$.
Therefore the sum of the delays corresponding to edges in $E_e$ is just $d_e$ which is expressed in
the following lemma.

\begin{lemma}
\label{lem:costconservation}
For every edge $e \in E$, and for every  $1 \leq t \leq n$,

$$\sum_{e' \in E_e}  d_{e'}(t) = d_e(t).$$

\end{lemma}

\begin{proof}

It is immediate from the definition that the delays satisfy for every $2 \leq k \leq n$, and $1 \leq t \leq n$,

$$ d_{e_{k}^+}(t) + d_{e_{k}^-}(t)= \left\{
    \begin{array}{ll}
        0  & \mbox{if } t<k, \\
        d_e(k)-d_e(k-1) & \mbox{si } t\geq k.
    \end{array}
\right.	$$
Therefore
\begin{eqnarray*}
\sum_{e' \in E_e}  d_{e'}(t) &=& d_e(1) +  \sum_{k=2}^t  ( d_{e_{k}^+}(t) + d_{e_{k}^-}(t) ) \\
&=&   d_e(1) + \sum_{k=2}^{t} ( d_e(k)-d_e(k-1) )\\
&=& d_e(t).
\end{eqnarray*}

\end{proof}
We now claim the following strong relationship between the cost functions in the two games.
\begin{lemma}
\label{lem:profilecostconservation}
For all state $s=(s_1, \ldots , s_n)$ in $G$, and for every player $i$, we have
	
	$$c_i(s') = c_i(s),$$
where $s' = (s_1', \ldots , s_n')$.
\end{lemma}
\begin{proof}
It is easy to verify the following sequence of equalities:
\begin{eqnarray*}
c_i(s') &=&  \sum_{e' \in s_i'} d_{e'}( f_{e'} (s')) \\
&=&   \sum_{e \in s_i}  \sum_{e' \in E_e} d_{e'}( f_{e'} (s')) \\
&=&   \sum_{e \in s_i}  \sum_{e' \in E_e} d_{e'}( f_{e} (s)) \\
&=&   \sum_{e \in s_i}   d_{e}( f_{e} (s)) \\
&=& c_i(s).
\end{eqnarray*}
Indeed the first and last equalities are just the definitions of the cost functions, and the second one is true
by the definition of $s_i'$. The third equality holds because for every edge $e \in E$, every $e' \in E_e$, 
and every player $i$,
the strategy $s_i$ contains $e$ if and only if $s'$ contains $e'$, and therefore $ f_{e'} (s') =  f_{e} (s)$.
The fourth eqality follows from Lemma~\ref{lem:costconservation}.
\end{proof}
By Lemma~\ref{lem:profilecostconservation}  we can
trivially deduce and $\epsilon$-approximate Nash equilibrium for $G$, 
given an $\epsilon$-approximate Nash equilibrium for $G'$. This concludes the proof of the theorem.
\end{proof}

\begin{theorem}
\label{theorem:flip}
For every polynomial time computable $0 < \eps < 1$,
computing an $\epsilon$-approximate Nash equilibrium in $n$-player
symmetric, flip congestion games is $\PLS$-hard.
\end{theorem}
\begin{proof}
The proof is very similar to the proof of the previous theorem. In the reduction the definition of the game $G'$ is as in
Theorem~\ref{theorem:non-alternating} except for the delay functions, for $2 \leq k \leq n$, which are defined now as 

$$ d_{e_{k}^+}(t) = \left\{
    \begin{array}{ll}
       ( d_e(k) - d_e(k-1) ) /2  & \mbox{if } t<k, \\
        (d_e(k)-d_e(k-1) ) /2 & \mbox{if } t\geq k,
    \end{array}
\right.	$$	
	and
	
$$ d_{e_{k}^-}(t) = \left\{
    \begin{array}{ll}
        -(d_e(k)-d_e(k-1)) /2  & \mbox{if } t<k, \\
        (d_e(k)-d_e(k-1)) /2 & \mbox{if } t\geq k.
    \end{array}
\right.	
$$
These are indeed flip functions, and it is easy to see that Lemma~\ref{lem:costconservation} holds again. This implies,
similarly to Theorem~\ref{theorem:non-alternating}, that the reduction is correct.
\end{proof}


\end{document}